\documentclass[12pt]{amsart}

\usepackage{amsmath}

\usepackage{amsfonts,amssymb,amsthm,enumitem}
\usepackage[margin=1.2in]{geometry}
\usepackage{mathtools}
\usepackage{nicefrac,setspace}
\usepackage{xcolor}
\usepackage{amsmath}

\usepackage{algorithm}
\usepackage{algpseudocode}
\usepackage{enumitem}
\usepackage{bbm}

\usepackage{enumitem}
\usepackage{hyperref}       
\usepackage{url}            
\usepackage{booktabs}       
\usepackage{amsfonts}       
\usepackage{nicefrac}       
\usepackage{microtype}      
\usepackage{comment}   
\usepackage{mathtools}

\makeatletter

\newcommand*\bigcdot{\mathpalette\bigcdot@{.5}}
\newcommand*\bigcdot@[2]{\mathbin{\vcenter{\hbox{\scalebox{#2}{$\m@th#1\odot$}}}}}
\makeatother

\newtheorem{theorem}{Theorem}
\newtheorem{claim}[theorem]{Claim}

\newtheorem*{theorem*}{Theorem}
\newtheorem{remark}{Remark}
\newtheorem*{claim*}{Claim}
\newtheorem*{remark*}{Remark}
\newtheorem*{lemma*}{Lemma}
\newtheorem{definition}{Definition}
\newtheorem{lemma}[theorem]{Lemma}
\newtheorem{fact}[theorem]{Fact}
\newtheorem*{fact*}{Fact}

\newcommand{\Z}{{\mathbb Z}}

\newcommand{\eps}{\epsilon}

\newcommand{\vspan}{\mathrm{span}}
\newcommand\restr[2]{{
  \left.
  #1 
  \right|_{#2} 
  }}

\newcommand{\F}{\mathbb{F}}
\newcommand{\E}{\mathop \mathbb{E}}

\title{List-Decoding Capacity Implies Capacity on the q-ary Symmetric Channel}

\begin{document}

\author[]{Francisco Pernice}
\address{Massachusetts Institute of Technology}
\email{fpernice@mit.edu}

\author[]{Oscar Sprumont}
\address{School of Computer Science, University of Washington}
\email{osprum@cs.washington.edu}

\author[]{Mary Wootters}
\address{Stanford}
\email{marykw@stanford.edu}

\begin{abstract}
    It is known that the Shannon capacity of the q-ary symmetric channel (qSC) is the same as the list-decoding capacity of an adversarial channel, raising the question of whether there is a formal (and black-box) connection between the two.
    We show that there is: Any linear code $C\subseteq \F_q^n$ that has minimum distance $d_{\min}=\omega\big(q^3\big)$ and achieves list-decoding capacity also achieves capacity on the qSC.
\end{abstract}

\maketitle

\section{Introduction}\label{intro}
A linear code of length $n$ over a finite field $\mathbb{F}_q$ is a $\mathbb{F}_q$-linear subspace $C \subseteq \mathbb{F}_q^n$.  In coding theory, we think of a code $C$ as the set of encodings of possible messages.  If a sender wants to send a message to a receiver over a noisy channel, they choose the corresponding element $c \in C$ (called a \emph{codeword}), and transmit $c$ over the channel.  A receiver sees a noisy version of the codeword, $\tilde{c}$, and must recover $c$. Two primary goals in designing such a code $C$ are (a) low redundancy, meaning that $n$, the length of a codeword $c \in C$, is not too much larger than $\log_q|C|$, the length of a message encoded with the code; and (b) tolerance of as many errors as possible.

The requirement of low redundancy is quantified by the \emph{rate} of the code: The rate $R$ of a code $C \subseteq \mathbb{F}_q^n$ is defined as $R = \frac{\log_q|C|}{n}$. The rate $R$ is always between $0$ and $1$, and the larger it is, the lower the redundancy of the code. 

The requirement on error tolerance depends on the channel model.  In this paper, we focus on two well-studied channel models, both parameterized by $p \in (0,1)$.  The first model is the $q$-ary symmetric channel\footnote{In the qSC$_p$, each $q$-ary symbol is corrupted independently with probability $p$; if a symbol is corrupted, it is replaced by a uniformly random different symbol in $\F_q$.} qSC$_p$. The second model is an adversarial channel that may corrupt up to a $p$-fraction of the symbols sent in a worst-case way.  Thus, in each model, we are concerned with the best possible trade-off between $R$ and~$p$.

In the case of the qSC$_p$, the best trade-off between $R$ and $p$ is well-understood.  The best rate at which reliable communication possible on the qSC$_p$---known as the \emph{Shannon capacity} of the channel---is $R = 1 - h_q(p)$, where 
\[ h_q(p) := (1 - p) \log_q\frac{1}{1-p} + p \log_q \frac{q-1}{p}\]
is the $q$-ary entropy function~\cite{shannon1948entropy}.  A family of codes that approaches this trade-off is said to \emph{achieve capacity} on the qSC$_p$ (see Definition \ref{defncapacityqsc}).

As one would expect, the best possible trade-off between $R$ and $p$ is worse in the adversarial case than in the stochastic case.\footnote{In fact, for small $q$, the best possible trade-off between $R$ and $p$ is still unknown in the adversarial model.}  However, if one relaxes the definition of ``reliable decoding,'' one can do better.  More precisely, we consider the notion of \emph{list-decoding}, a classical notion introduced by Elias and Wozencraft~\cite{elias1957firstlist,wozencraft1958firstlist}.  In list-decoding, the decoder's goal is no longer to return only the transmitted codeword $c \in C$, but rather a short list of possible codewords that is guaranteed to include $c$.  Formally, $C$ is \emph{$(p,L)$-list-decodable} if for all $w \in \F_q^n$, 
\[ |\{ c \in C \,:\, d(w,c) \leq pn \} | \leq L.\]
It is known (see, e.g., \cite{2023bookcodingtheory}, Theorem 7.4.1) that the best
\footnote{In this informal discussion, we have not mentioned the size $L$ of the list the decoder is allowed to output. 
Slightly more formally, the \emph{list-decoding capacity theorem}  states that there are codes of rate $R= 1 - h_q(p) - \eps$ that are $(p,L)$-list-decodable where $L$ depends on $\eps$ but not on $n$; and that conversely any code of rate bounded above $1 - h_q(p)$ cannot be list-decodable with any list size $L$ that is sub-exponential in $n$.  } 
possible trade-off between $R$ and $p$ in the list-decoding setting is also $1 - h_q(p)$, exactly the same as the Shannon capacity of the qSC$_p$!  If a family of list-decodable codes approaches this trade-off, we say that it \emph{achieves list-decoding capacity} (see Definition \ref{defnlistcapacity}).

\vspace{.3cm}

\paragraph{\textbf{Our question.}}
This state of affairs raises a natural question: Since the capacity of the qSC$_p$ is the same as the list-decoding capacity, is there a formal connection between these two notions?  
In the list-decoding literature, it is common to introduce list-decoding as a ``bridge'' between the Shannon and Hamming models, in that list-decoding allows one to reach the Shannon capacity of a channel, even under a worst-case (Hamming) model of errors.\footnote{For example, the chapter on list-decoding in the textbook~\cite{2023bookcodingtheory} is titled ``Bridging the Gap Between Shannon and
Hamming: List Decoding.''} But can this be made formal?

In this paper we show that there is a formal connection, at least in one direction.  That is, we show that \emph{any} list-decodable code $C$ with good enough minimum distance\footnote{The \emph{minimum distance} of a code $C$ is $d_{\min}(C) = \min_{c \neq c' \in C}d(c,c')$, where $d(\cdot,\cdot)$ denotes Hamming distance.} achieves capacity on the qSC$_p$.  (As we note below, the converse is not true.)  In the next section, we describe our results in more detail.

\subsection{Our Results}\label{sec:results}
Our main result is the following theorem.
\begin{theorem}\label{capacitybridge}
Let $p \in (0,1)$.
    Let $\{C_n\subseteq\F_{q}^n\}$ be a family of linear codes that achieves list-decoding capacity on the adversarial channel that introduces a $p$-fraction of corruptions. If $d_{\min}(C_n)=\omega\Big(\frac{q^3}{(1-p)^2}\Big)$, then $\{C_n\}$ achieves capacity on the qSC$_p$.
\end{theorem}

Theorem~\ref{capacitybridge} follows from  a more general statement about $(p,L)$-list-decodability, which we present in Theorem \ref{mainresultformal}. Next, we state a slightly weaker but more easily digestible version of that result.

\begin{theorem}\label{mainresult}
Let $C\subseteq \F_q^n$ be a linear, $(p,L)$-list-decodable code with minimum distance $d_{\min}(C)\geq \frac{q^3}{(1-p)^2} \ln^4(nL)$ for sufficiently large $n$. Then $C$ can be used for reliable communication on the qSC$_{p'}$, where ${p}' = p-\frac{7}{\ln n}$.
\end{theorem}

See Section \ref{sectionmainresult} for the proof of Theorems \ref{capacitybridge} and \ref{mainresult}. 

\begin{remark}[Requirement on the field size]\label{rem:fieldsize}
    As stated, the hypothesis of Theorem \ref{mainresult} requires that $d_{\min}(C)$ be at least $q^3$, which in particular implies that $q < n^{1/3}$. This rules out codes like Folded Reed-Solomon codes, where $q$ is a large polynomial in $n$.  However, for known constructions of list-decodable codes over growing alphabets, the conclusion of Theorem~\ref{mainresult} already follows from previous work, and so our contribution is to focus on the case where $q \geq 2$ is fixed.  
    
    In more detail, it is known that any linear code with distance at least $pn$ over a sufficiently large alphabet admits reliable communication over the qSC$_{p'}$ for some $p' = p - o(1)$~\cite{RU10}; and it is also known that any linear $(p,L)$-list-decodable code $C \subseteq \F_q^n$ with $q \geq L$ has distance at least $pn$~\cite[Proposition 6.5]{goldberg2024ListToDistance}.  Thus, any linear $(p,L)$-list-decodable code $C \subseteq \F_q^n$ with $q = \omega(1)$ and $L=O(1)$ has distance at least $pn$ by \cite{goldberg2024ListToDistance}, and therefore admits reliable communication on the qSC by \cite{RU10}.
\end{remark}

We note that the converse of Theorem \ref{capacitybridge} (qSC capacity implying list-decoding capacity) does not hold in general. For instance, Reed-Muller codes achieve capacity over the BSC$_p$ \cite{reeves2021bitcapacity,abbe2023rmcapacityBSC} but have $2^{\Omega(n)}$ codewords of weight $\leq pn$, for every constant $p$ (see, e.g., \cite{abbe2023survey2}). We leave it as an intriguing open problem to identify sufficient conditions under which the converse holds. 

We show in Appendix~\ref{app:distance} that the requirement that the minimum distance be large cannot be avoided entirely.  More precisely, we show that there exist codes $C\subseteq\F_q^n$ with constant minimum distance that achieve list decoding capacity but do not achieve capacity on the $q$-ary symmetric channel.

Finally, we give a simple proof that Theorem \ref{mainresult} holds for erasures, without the linearity and alphabet size requirements.
\begin{theorem}\label{capacitythm}
    Let $C\subseteq\Z_q^n$ be a $(p,L)$-list-decodable code with minimum distance $d_{\min}(C) = \omega(\log L).$ Then $C$ can be used for reliable communication on the q-ary erasure channel, qEC$_{p'}$, where $p' = p - o(1)$.
\end{theorem}
We prove Theorem~\ref{capacitythm} in Appendix \ref{aerasure}.
\vspace{.3cm}

\begin{remark}[New results for capacity-achieving codes on the qSC$_p$?]
It is natural to ask whether or not our results imply new constructions of capacity-achieving codes on the qSC$_p$.  The answer is yes, in that there are known capacity-achieving list-decodable codes over constant-sized alphabets, even with efficient algorithms, to which our theorem applies and which to the best of our knowledge had not previously been known to achieve capacity on the qSC (e.g.~\cite{Hemenway2017AGcodes2,Kopparty2018multiplicitycodes3,Guo2022AGcodes3}).  However, it seems unlikely that these constructions---which are designed to solve the (as we show here) strictly harder problem of list-decoding---would yield more practical results on the qSC than polar codes~\cite{arikan2009polar}.  Thus, our main contribution is to establish a formal connection between these two problems.
\end{remark}

\subsection{Overview of Techniques}\label{sec:tech}
Before we discuss our techniques, we set up a bit more notation.  For a transmitted codeword $c \in C$, we write the received corrupted codeword as \[ \tilde{c} = c + z,\]
where $z \in \F_q^n$ is an error vector.  Thus, in the qSC$_p$, each coordinate of $z_i$ is independent, and is equal to $0$ with probability $1-p$ and is uniformly random in $\F_q^*$ with probability $p$.  In the adversarial channel, $z$ is an arbitrary vector of weight at most $pn$.

Now we can give a high-level overview of our techniques.
Let $D^*:\F_q^n\rightarrow C$ be a maximum-likelihood decoder on the qSC$_p$.\footnote{That is, given a corrupted codeword $\tilde{c} = c + z$, the decoder $D^*$ returns a closest codeword $D^*(\tilde{c})$ to $\tilde{c}$.}  As $C$ is linear, we can pick $D^*$ so that its success depends only on the error vector $z$, and we may assume without loss of generality that the transmitted codeword was $c=\vec{0}$.  Define the function
\begin{align}\label{defnfintro}
    f(z):=\begin{cases}
1 & \text{if $D^*(z)=\vec{0}$},\\
0 & \text{otherwise}.
\end{cases}
\end{align}
The expectation of $f$ is exactly the probability that $D^*$ outputs the correct codeword. Showing that $C$ allows reliable communication on qSC$_{p'}$ for some $p' = p-o(1)$ is thus equivalent to showing that
\begin{align}\label{endpointtrans}
    \E_{z\sim p'}[f(z)]=1-o(1).
\end{align}
In the equation above, we use $z\sim p'$ to mean that, independently for each coordinate $i\in[n]$,
\[ z_i = \begin{cases} 0 & \text{ with probability } 1 - p' \\ j \in \{1, 2, \ldots, q-1\} & \text{ with probability }  \frac{p'}{q-1} \end{cases}.\]  One key observation is that if $C$ is $(p,L)$-list-decodable, then we must have
\begin{align}\label{startingpoint}
    \E_{z\sim p'}[f(z)]\geq \frac{1}{L}-o(1)
\end{align}
for, say, $p' = p - n^{-1/4}$.
Indeed, consider the following decoder $D$: upon receiving some corrupted codeword $m\in\F_q^n$, find all the codewords $c\in C$ such that wt$(c+m)\leq pn$, and output one such codeword uniformly at random. Since $C$ is $(p,L)$-list decodable, there can never be more than $L$ codewords $c\in C$ satisfying wt$(c+m)\leq pn.$ Thus as long as the error string $z$ has weight smaller than $pn$, the decoder $D$ will succeed in outputting the correct codeword with probability at least $\frac{1}{L}$. But if $z \sim p'$, then $z$ has weight smaller than $pn$ with high probability.
Thus, on $z\sim p'$ our decoder $D$ outputs the correct codeword with probability at least $\frac{1}{L}-o(1)$. Since the max-likelihood decoder $D^*$ is optimal, it must perform at least as well as the decoder $D$, and we thus get (\ref{startingpoint}).

In order to deduce (\ref{endpointtrans}) from (\ref{startingpoint}), it will then suffice to show that the function $g(p')=\E_{z\sim p'}[f(z)]$ has a sharp transition as a function of $p'$. This was proven for $q=2$ in \cite{zemor93application,tillich2000applications2}.  Thus, combined with their result, the above argument immediately implies Theorem~\ref{mainresult} for binary codes.  Our main technical contribution is to generalize their arguments to larger field sizes.\footnote{As discussed in Section~\ref{sec:related}, such a generalization was asserted in~\cite{kindarji2010generalization} in a different context, but as discussed in Appendix~\ref{appendixmistake}, we believe that their proof is missing key elements.} Formally, our goal will be to bound the derivative of $\E[f]$ by 
\begin{align}\label{goalintro}
\frac{d}{d p}\E_{z\sim p}[f(z)]\leq-\omega(1) \cdot \E_{z\sim p}[f(z)]\big(1-\E_{z\sim p}[f(z)]\big)
.
\end{align}
Margulis \cite{margulis1974transition} and Russo \cite{russo1982transition} pioneered the use of such inequalities for proving that the expectation of certain Boolean functions transitions quickly from $1-o(1)$ to $o(1)$.  The point is that whenever $\E[f]$ is away from $0$ and away from $1$, we have $\E[f](1-\E[f])$ away from $0$, and thus the $\omega(1)$ term in (\ref{goalintro}) ensures that in this regime the derivative of $\E[f]$ is large. 

Now for any monotone function $f:\F_q^n\rightarrow\{0,1\}$, we can bound the derivative of the expectation as 
        \begin{align}\label{russointro}
        \frac{d}{dp}\E_{z\sim p}[f(z)]\leq-\frac{1}{q-1}\E_{z\sim p}[h_f(z)],
    \end{align}
where we define the quantity  
\begin{align*}
    h_f(z):=\begin{cases}
\Big| \big\{i\in[n]:z_i=0 \textnormal{ and }\exists a\neq0 \textnormal{ s.t. }f(za_i)=0   \big\} \Big| & \text{if $f(z)=1$},\\
0 & \text{otherwise}.
\end{cases}
\end{align*}

For $q=2$, the inequality (\ref{russointro}) first appeared in \cite{margulis1974transition,russo1982transition} and is called Russo's Lemma; the generalization to larger $q$ is stated as Lemma~\ref{russoFq}. To make use of (\ref{russointro}), we prove in Section \ref{sectionisoperimetry} the following isoperimetric inequality, which is a generalization of a bound proven by Talagrand for $q=2$:
\begin{align}\label{isointro}
    \E_{z\sim p}\big[h_f(z)\big]&\geq \frac{1-p}{2}\sqrt{\Delta_f}\cdot \E_{z\sim p}[f(z)]\big(1-\E_{z\sim p}[f(z)]\big),
\end{align}
where we denoted the minimum positive value of $h_f$ by
\begin{align*}
    \Delta_f:=\min_{x\in\F_q^n:h_f(x)\neq 0}\{h_f(x)\}.
\end{align*}
To obtain our desired inequality (\ref{goalintro}) from the bounds (\ref{russointro}) and (\ref{isointro}), it will thus suffice to show that our specific function $f$ (the indicator function of a successful decoding, defined in (\ref{defnfintro})) satisfies $\Delta_f=\omega(1)$. That is, we want to show that if $x\in\F_q^n$ is some error string that leads to correct decoding, and if one of the neighbors of $x$ leads to incorrect decoding, then there must be many such ``bad'' neighbors of~$x$. 

Intuitively, this is because in order for $x$ and one of its neighbors to be mapped to different codewords, it must be the case that $x$ is about halfway between the transmitted codeword (as above, let us assume that the transmitted codeword is the all-$0$ codeword), and some other codeword $c.$ Formally, it must be the case that
\begin{align*}
    d(x,0)\leq d(x,c)\leq d(x,0)+2,
\end{align*}
where $d(a,b)$ denotes the Hamming distance between vectors $a$ and $b$.
For simplicity, in this section assume that $d(x,c)=d(x,0)+1$. Then for any coordinate $i\in[n]$ where $x_i=0$ and $c_i\neq 0$, the neighbor $x'$ of $x$ obtained by setting the $i^\textnormal{th}$ coordinate to $c_i$ is closer to $c$ than to $0$. Each such neighbor $x'$ thus satisfies $f(x')=0$, and we have
\begin{align*}
    h_f(x)\geq \big|\big\{  i\in[n]:x_i=0,c_i\neq 0 \big\}\big|.
\end{align*}
The exact number of such coordinates $i\in[n]$ can be bounded in terms of the distance between $0$ and $c$, which itself is bounded by the minimum distance of the code. See Section \ref{sectionsharptransition} for more details.

Once we obtain a lower bound on $\Delta_f$, our desired inequality (\ref{goalintro}) follows from equations (\ref{russointro}) and (\ref{isointro}). 

\subsection{Related Work}\label{sec:related}
Both capacity-achieving codes on the qSC and capacity-achieving list-decodable codes are extremely well-studied, with lines of work going back to the 1950's or earlier.  In this section, we mention a few of the most relevant works.
\vspace{.5cm}

\paragraph{\textbf{Capacity-Achieving Codes on the qSC$_p$.}}
As noted above, the capacity of the qSC$_p$ has been known since Shannon's seminal work in the 1940's~\cite{shannon1948entropy}.  Shannon already observed that random codes achieve capacity, but it was not until the 1960's that explicit constructions were obtained with Forney's \emph{concatenated codes}~\cite{forney1966concatenated}.  These codes have poor performance in terms of the gap to capacity, which has been more recently improved by Ar{\i}kan with polar codes~\cite{arikan2009polar}; polar codes also have efficient decoding algorithms.  Even more recently, Reed-Muller codes have also been shown to achieve capacity on the binary symmetric channel~\cite{abbe2015RMlowrate,kudekar2016erasure,reeves2021bitcapacity,abbe2023rmcapacityBSC}.  Ensembles of LDPC codes are also known to achieve capacity under maximum-likelihood decoding~\cite{gallager1962ldpc}, or even under efficient algorithms~\cite{luby1997ldpc2,kudekar2013ldpc}. 

\vspace{0.5cm}
\paragraph{\textbf{List decoding.}}
The notion of list decoding was first introduced by Elias \cite{elias1957firstlist} and Wozencraft \cite{wozencraft1958firstlist}, and since then several ensembles of codes have been shown to achieve list-decoding capacity.  Here, we survey these codes and also comment on how Theorem~\ref{capacitybridge} applies to them.

The original work of Elias and Wozencraft showed that uniformly random codes achieve list-decoding capacity with high probability, and a more recent line of work has established the same for uniformly random linear codes \cite{zyablov1981list,Guruswami2010randomlinear1,wootters2013randomlinear2,Li2021randomlinear3,Guruswami2022randomlinear4}.  Both random codes and random linear codes are known to achieve capacity on the qSC with high probability, so our results do not imply anything new for these ensembles.

The first explicit constructions were obtained by Guruswami and Rudra, who showed in \cite{Guruswami2008FoldedRS} that \emph{folded Reed-Solomon codes} achieve list decoding capacity.  
This breakthrough was followed by a line of work on explicit constructions, both improving the list size for folded RS codes (e.g.,~\cite{Kopparty2018multiplicitycodes3,tamo2024tighterlists,S24,CZ24}), and by extending these results to other families of codes, including multiplicity codes~\cite{Guruswami2013multiplicitycodes1,Kopparty2015multiplicitycodes2}.  
Going back to randomized constructions (though with more structure than random linear codes), a recent line of work has also established that randomly punctured Reed-Solomon codes \cite{Brakensiek2023puncturedRS1,Guo2023puncturedRS2,Alrabiah2024puncturedRS3,Berman2024puncturedRS4} are also capacity-achieving list-decodable codes.
As discussed in Remark~\ref{rem:fieldsize}, our resuls to not apply to these codes (as they are over large alphabets), but for those with constant list-sizes, the fact that they achieve capacity on the qSC follows from known results.

Finally, we turn to constructions of capacity-achieving list-decodable codes with constant alphabet sizes.  These include Gallager's ensemble of
LDPC codes \cite{Mosheiff2020ldpccapacity}; concatenated codes~\cite{GRconcat}; and constructions based on algebraic geometry (AG) codes, for example~\cite{Guruswami2013AGcodes1,Hemenway2017AGcodes2,Guo2022AGcodes3,Guruswami2022AGcodes4,Brakensiek2023AGcodes5}. 
Some of these codes---for example, \cite{Hemenway2017AGcodes2, Guo2022AGcodes3}---come with efficient list-decoding algorithms and were also not already known to achieve capacity on the qSC$_p$.  Combined with our results, this implies that not only do these codes achieve capacity on the qSC$_p$, but they also have efficient decoding algorithms. 
\vspace{0.5cm}

\paragraph{\textbf{Threshold Phenomena.}}  
As mentioned in Section~\ref{sec:tech}, our main technical contribution is a threshold result for the success probability of decoding on the qSC$_p$.  Our techniques build on a long line of work, which we briefly mention here.

 Margulis \cite{margulis1974transition}, Russo \cite{russo1982transition} and Talagrand \cite{Talagrand93} showed that the expectation of any monotone\footnote{See Section \ref{defnmonotone} for a formal definition of monotonicity.} Boolean function $f$ sharply transitions from $\E[f]\geq 1-o(1)$ to $\E[f]\leq o(1)$, as long as no $z\in f^{-1}(1)$ has a small, nonzero number of neighbors in $f^{-1}(0)$. This fact has already seen many applications in coding theory \cite{zemor93application,tillich2000applications2,tillich2004applications3,kudekar2016erasure,kumar2016applications4}. In particular, Tillich and Z{\'e}mor proved in \cite{zemor93application,tillich2000applications2} that the decoding probability of any binary linear code with large minimum distance undergoes a sharp transition. Our main technical contribution is a generalization of these results to larger alphabets (see Sections \ref{sectionisoperimetry} and \ref{sectionsharptransition}, and in particular Theorem \ref{gbound}).

We note that an attempt was already made in \cite{kindarji2010generalization} to generalize the results of Tillich and Z{\'e}mor to larger alphabets, but unfortunately the proof in~\cite{kindarji2010generalization} seems to be missing key elements and does not seem to be correct as written.  We explain in Appendix~\ref{appendixmistake} why we think that a new proof of this generalization is needed.

\section{Conventions and Preliminaries}
\subsection{Finite fields}
We will work with the finite field $\F_q$ over $q$ elements. Given any vector $z\in\F_q^n$, we denote its Hamming weight by
\begin{align*}
    \textnormal{wt}(z):=\big\{i\in[n]:z_i\neq 0\big\}.
\end{align*}
Given two vectors $y,z\in\F_q^n$, we denote their Hamming distance by
\begin{align*}
    d(y,z):=\textnormal{wt}(z-y).
\end{align*}

\subsection{Probability Theory}

For any probability distribution $\mathcal{D}$ over a set $X$, we will use the notation $x\sim \mathcal{D}$ to denote a random element $x\in X$ sampled according to $ \mathcal{D}.$ The main probability distribution that we will use throughout this paper is the distribution of a \emph{$p$-noisy string} over $\F_q^n$, which we define as follows.  For $p \in [0,1]$, we use  
$$z\sim p$$
to denote a $p$-noisy string $z \in \F_q^n$, meaning that for each $i \in \{1, \ldots, n\}$, $i^\textnormal{th}$ entry is independently
\begin{align*}
    z_i&=\begin{cases}
0 & \text{with probability $1-p$},\\
j& \text{with probability $\frac{p}{q-1}$, for all $j\in\{1,2,\dotsc,q-1\}$}.
\end{cases}
\end{align*} 

We will also need Hoeffding's inequality (see for example, \cite{concentrationinequalities}).

\begin{lemma}[Hoeffding's Inequality]\label{hoeffding}
    Let $X_1,X_2,\dotsc,X_n$ be  independent random variables taking values in $[0,1]$. Then for any $t>0$, we have
    \begin{align*}
    \Pr\Big[\sum_{i=1}^nX_i>\sum_{i=1}^n\E[X_i]+t\Big]\leq e^{-\frac{2t^2}{n}}
    \end{align*}
    and
    \begin{align*}
    \Pr\Big[\sum_{i=1}^nX_i<\sum_{i=1}^n\E[X_i]-t\Big]\leq e^{-\frac{2t^2}{n}}.
    \end{align*}
\end{lemma}

\subsection{Monotonicity}\label{defnmonotone}
A property that will play a key role in our analysis is monotonicity. We say that a function $f:\F_q^n\rightarrow \{0,1\}$ is \emph{monotone decreasing} if for any index $i\in[n]$, any point $z\in \F_q^{n}$, and any $a\in\{1,2,\dotsc,q-1\}$, we have
\begin{align*}
    f(z0_i)\leq f(za_i).
\end{align*}
The expectation of any monotone decreasing function is decreasing (see for e.g. \cite{concentrationinequalities}, page 280).
\begin{fact}\label{strictlyincreasing}
    For any monotone decreasing function $f:\F_q^n\rightarrow\{0,1\}$, the function
    \begin{align*}
        g(p):=\E_{z\sim p}[f(z)]
    \end{align*}
    is decreasing. 
\end{fact}
We derive some useful properties of monotone functions in Section \ref{sectionisoperimetry}.

\subsection{Coding and Decoding}
A linear code of length $n$ over the field $\F_q$ is a subspace $C\subseteq\F_q^n$. 
A deterministic decoder for a code $C\subseteq\F_q^n$ is a function $D:\F_q^n\rightarrow C$, and a randomized decoder is a probability distribution over the set of deterministic decoders. We say that the code $C\subseteq \F_q^n$ \emph{admits reliable communication} on the qSC$_p$ if there exists a randomized decoder $D$ such that for all $c\in C$, we have 
\begin{align*}
    \Pr\big[D(c+z)=c\big]\geq 1-o(1),
\end{align*}
where the probability is taken over both the randomized decoder $D$ and the random variable $z\sim p$.
A \emph{maximum-likelihood} decoder $D^*$ over the qSC is a deterministic decoder that, upon seeing a corrupted message $m$, returns a codeword $c\in C$ closest to $m$. We say that $D^*$ is \emph{symmetric} if it breaks ties in a symmetric way, meaning that if $D^*(z)=0$ then $D^*(z+c)=c$. It is well-known that any symmetric maximum-likelihood decoder is optimal for decoding random errors (see, e.g. \cite{1977bookkrawtchouk}, page 8).
\begin{fact}\label{bestdecoder}
    For any code $C\subseteq\F_q^n$, any corresponding symmetric maximum-likelihood decoder $D^*$ is optimal. That is, for any $p\in(0,1)$ and any decoder $D$, we have
    \begin{align*}
        \min_{c\in C}\Big\{ \Pr_{\substack{z\sim p}}\big[D(c+z)=c\big] \Big\}\leq \min_{c\in C}\Big\{ \Pr_{\substack{z\sim p}}\big[D^*(c+z)=c\big] \Big\}.
    \end{align*}
\end{fact}

\subsection{Communication Channels}
We will be interested in the performance of codes under both stochastic and adversarial noise. 

The model we will use for stochastic noise is the q-ary Symmetric Channel, which is characterized by a parameter $p\in(0,1)$. 
When an element $c$ of the code $C\subseteq\F_q^n$ is sent through the q-ary Symmetric Channel qSC$_p$, each of its $n$ entries is corrupted independently at random with probability $p$. If corrupted, the $i^\textnormal{th}$ entry $c_i$ is replaced by a uniformly random element in $\F_q\setminus\{c_i\}$. 
It is well-known \cite{shannon1948entropy} that a uniformly random code $C\subseteq\F_q^n$ of rate 
\begin{align*}
    \textnormal{rate}(C)=1-h_q(p)-o(1)
\end{align*}
will with high probability admit reliable communication on the qS$C_p$, where $h_q$ denotes the q-ary entropy function
\[ h_q(p) := (1 - p) \log_q\frac{1}{1-p} + p \log_q \frac{q-1}{p}.\]
On the other hand, no code $C\subseteq\F_q^n$ of rate $1-h_q(p)+o(1)$ can admit reliable communication on the qS$C_p.$ Thus we have the following definition of qSC-capacity. 
\begin{definition}\label{defncapacityqsc}
A sequence of codes $\{C_n\subseteq\F_q^n\}$ of rate $1-h_q(p)$ achieves capacity over the qS$C_p$ if there exists a function $\epsilon(n)=o(1)$ such that each $C_n$ satisfies
\begin{align*}
    \Pr_{z\sim p-\epsilon_n}[D(c+z)=c]\geq 1-\epsilon_n.
\end{align*}
\end{definition}

In the adversarial noise model, the location of the errors is decided by some adversary rather than being stochastic. We say that the code $C\subseteq\F_q^n$ is \emph{$(p,L)$-list decodable} if  for any $z\in\F_q^n$, 
the ball of radius $pn$ around $z$ contains at most $pn$ codewords $c\in C.$
A uniformly random code $C\subseteq \F_q^n$ of rate
\begin{align}\label{capacitylist}
    \textnormal{rate}(C)=1-h_q(p)-\epsilon
\end{align}
is $(p,O(\frac{1}{\epsilon}))$-list decodable with high probability (see for e.g. \cite{2023bookcodingtheory}, Theorem 7.4.1). On the other hand, no code $C\subseteq \F_q^n$ of rate larger than $1-h_q(p)+\epsilon$ is $(p,L)$-list decodable for any $L<q^{\Omega(\epsilon n)}$. Thus we obtain the following definition of list decoding capacity.
\begin{definition}\label{defnlistcapacity}
    A sequence of codes $\{C_n\subseteq\F_q^n\}$ of rate $1-h_q(p)$ achieves list decoding capacity if for every function $L(n)=\omega(1)$, there exists a function $\epsilon(n)=o(1)$ such that each $C_n$ is $(p-\epsilon(n), L(n))$-list decodable.
\end{definition}
 We note for example that by equation (\ref{capacitylist}), uniformly random codes achieve list decoding capacity with $\epsilon(n)=O(\frac{1}{L(n)}).$ But in general, we allow for worse dependence of $\epsilon$ on $L$: we only require that $\epsilon$ go to $0$ as $n$ goes to infinity.

\section{An isoperimetric inequality over finite fields}\label{sectionisoperimetry}
In this section, we generalize an $\F_2$-result of Talagrand to finite fields of all sizes. For any function $f:\F_q^n\rightarrow \{0,1\}$ and any point $z\in\F_q^n$, we define the quantity
\begin{align*}
    h_f(z):=\begin{cases}
\Big| \big\{i\in[n]:z_i=0 \textnormal{ and }\exists a\neq0 \textnormal{ s.t. }f(za_i)=0   \big\} \Big| & \text{if $f(z)=1$},\\
0 & \text{otherwise}.
\end{cases}
\end{align*}
Our goal will be to relate the quantities $\E[h_f]$ and $\E[f].$ Theorem~\ref{thmtalagrand} below was proven for the special case of $q=2$ by Talagrand in \cite{Talagrand93}; we extend that result to arbitrary field sizes.
\begin{theorem}\label{thmtalagrand}
        For any monotone decreasing function $f:\F_q^n\rightarrow\F_2$ and any noise parameter $p\in[0,1]$, we have
        \begin{align*}
            \E\Big[ \sqrt{h_f} \Big]\geq \frac{1-p}{2}\cdot \E[f]\big(1-\E[f]\big),
        \end{align*}
        where all the expectations are taken over the q-ary $p$-noisy distribution.
    \end{theorem}
    \begin{proof}
        We proceed by induction on $n$. For the base case $n=1$, either we have $f(0)=f(a)$ for all $a\in\F_q$, in which case the right-hand side is $0$ and the inequality holds trivially; or we have $f(0)=1$ and $f(a)=0$ for some $a\in\{1,2,\dotsc,q-1\}$, in which case we get
        \begin{align*}
            \E\Big[ \sqrt{h_f} \Big]&=1-p\\
            &\geq \frac{1-p}{2}\cdot  \E[f]\big(1-\E[f]\big).
        \end{align*}
        We thus turn to the induction step. Suppose the desired statement holds for $n-1$, and consider some function $f:\F_q^n\rightarrow\F_2.$ We may assume that
        \begin{align}\label{assumeEf}
             \E\Big[ \sqrt{h_f} \Big]\leq \frac{1-p}{2},
        \end{align}
        as otherwise the desired claim is trivial. For each 
$a\in\F_q$, define the following function of $n-1$ variables. 
\begin{align*}
    f_a(x):=f(xa_n).
\end{align*}
For convenience, we will denote the expectation of each of these functions by 
$E_a:=\E[f_a]$, where again the expectations are taken over the $p$-noisy distribution. By definition, we have
\begin{align}
    \E[f]=(1-p)E_0+\frac{p}{q-1}\sum_{a\in\F_q\setminus \{0\}}E_a,
\end{align}
and thus
\begin{align*}
    \E[f]\big(1-\E[f]\big)&=\Big((1-p)E_0+\frac{p}{q-1}\sum_{a\in\F_q\setminus \{0\}}E_a  \Big)\Big(1-(1-p)E_0-\frac{p}{q-1}\sum_{a\in\F_q\setminus \{0\}}E_a\Big)\nonumber\\
    &=(1-p)E_0\Big(1-E_0+pE_0-\frac{p}{q-1}\sum_{a\in\F_q\setminus \{0\}}E_a)\Big)\\
    &\quad+\frac{p}{q-1}\sum_{a\in\F_q\setminus \{0\}}E_a\Big(  1-E_a+(1-\frac{p}{q-1})E_a-(1-p)E_0-\frac{p}{q-1}\sum_{b\in\F_q\setminus\{0,a\}}E_b\Big).
\end{align*}
Extracting from the expression above the terms corresponding to the variance of each $f_a$, we get
\begin{align}\label{expandvarnew}
     \E[f]\big(1-\E[f]\big)&=(1-p)E_0(1-E_0)+\frac{p}{q-1}\sum_{a\in\F_q\setminus \{0\}}E_a(1-E_a)\nonumber\\
     &\quad +(1-p)pE_0\Big(E_0-\frac{1}{q-1}\sum_{a\in\F_q\setminus \{0\}}E_a  \Big)\nonumber\\
     &\quad+\frac{p}{q-1}\sum_{a\in\F_q\setminus \{0\}}E_a\Big((1-\frac{p}{q-1})E_a-(1-p)E_0\Big)\nonumber\\
     &\quad-\frac{p^2}{(q-1)^2}\sum_{\substack{a\in\F_q\setminus \{0\}\\b\in\F_q\setminus\{0,a\}}}E_aE_b.
\end{align}
The first line in the equation above is the sum of the individual variances. We will now want to bound the contribution of the other terms. For this it will be useful to replace each factor of $1-p$ by a factor of $1-\frac{p}{q-1}$, so that we can complete the square. That is, we write the summands in the two middle lines of (\ref{expandvarnew}) as
\begin{align*}
    (1-p)pE_0\Big(E_0-\frac{1}{q-1}\sum_{a\in\F_q\setminus \{0\}}E_a  \Big)
    &=(1-\frac{p}{q-1})\frac{p}{q-1}\sum_{a\in\F_q\setminus \{0\}}E_0\big(E_0-E_a  \big)\\
    &\quad-(1-\frac{1}{q-1})\frac{p^2}{q-1}\sum_{a\in\F_q\setminus \{0\}}E_0(E_0-E_a)
\end{align*}
and
\begin{align*}
    \frac{p}{q-1}\sum_{a\in\F_q\setminus \{0\}}E_a\Big((1-\frac{p}{q-1})E_a-(1-p)E_0\Big)&=\frac{p}{q-1}(1-\frac{p}{q-1})\sum_{a\in\F_q\setminus \{0\}}E_a\Big(E_a-E_0\Big)\\
    &\quad +\sum_{a\in\F_q\setminus \{0\}}\frac{p^2}{q-1}(1-\frac{1}{q-1})E_aE_0.
\end{align*}
Combining the two equations above with (\ref{expandvarnew}), we get
\begin{align*}
     \E[f]\big(1-\E[f]\big)&=(1-p)E_0(1-E_0)+\frac{p}{q-1}\sum_{a\in\F_q\setminus \{0\}}E_a(1-E_a)\nonumber\\
     &\quad +(1-\frac{p}{q-1})\frac{p}{q-1}\sum_{a\in\F_q\setminus \{0\}}\big(E_0-E_a\big)^2-\frac{p^2(q-2)}{(q-1)^2}\sum_{a\in\F_q\setminus \{0\}}E_0(E_0-E_a)\nonumber\\
     &\quad +\sum_{b\in\F_q\setminus\{0\}}\frac{p^2(q-2)}{(q-1)^2}E_bE_0-\frac{p^2}{(q-1)^2}\sum_{\substack{a\in\F_q\setminus \{0\}\\b\in\F_q\setminus\{0,a\}}}E_aE_b.
\end{align*}
Adding an artificial summation to the fourth and fifth sums above, we get
\begin{align*}
     \E[f]\big(1-\E[f]\big)&=(1-p)E_0(1-E_0)+\frac{p}{q-1}\sum_{a\in\F_q\setminus \{0\}}E_a(1-E_a)\nonumber\\
     &\quad +(1-\frac{p}{q-1})\frac{p}{q-1}\sum_{a\in\F_q\setminus \{0\}}\big(E_0-E_a\big)^2\nonumber\\
     &\quad -\frac{p^2}{(q-1)^2}\sum_{\substack{a\in\F_q\setminus \{0\}\\b\in\F_q\setminus\{0,a\}}}\Big(E_0(E_0-E_a)-E_bE_0+E_aE_b\Big).
\end{align*}
Defining the quantity $$E_{\min}:=\E_{x\sim p^{n-1}}\Big[\min_{a\in\F_q\setminus \{0\}}\big\{f(xa_n)\big\}\Big],$$ we can then bound the variance of $f$ by 
\begin{align}\label{expandvar}
     \E[f]\big(1-\E[f]\big)&\leq(1-p)E_0(1-E_0)+\frac{p}{q-1}\sum_{a\in\F_q\setminus \{0\}}E_a(1-E_a)+(1-\frac{p}{q-1})p\big(E_0-E_{\min}\big)^2\nonumber\\
     &\quad -\frac{p^2}{(q-1)^2}\sum_{\substack{a\in\F_q\setminus \{0\}\\b\in\F_q\setminus\{0,a\}}}(E_0-E_b)(E_0-E_a)\nonumber\\
     &\leq (1-p)E_0(1-E_0)+\frac{p}{q-1}\sum_{a\in\F_q\setminus \{0\}}E_a(1-E_a)+(1-\frac{p}{q-1})p\big(E_0-E_{\min}\big)^2.
\end{align}
Now that we have obtained a convenient expression for the right-hand side of our theorem's inequality, we turn to bounding the left-hand side. We note that since $f$ is monotone, we must have
\begin{align*}
    \E_{z\sim p^n}\Big[ \sqrt{h_f(z)} \Big]&=(1-p)\E_{x\sim p^{n-1}}\Big[ \sqrt{h_{f_0}(x)+\mathbbm{1}\big\{f(x0_n)=1,f(xa_n)=0 \textnormal{ for some }a\in\F_q\setminus \{0\}\big\}} \Big]\\
    &+\frac{p}{q-1}\sum_{\substack{a\in\F_q\setminus \{0\}}}\E_{x\sim p^{n-1}}\Big[ \sqrt{h_{f_a}(x)} \Big].
\end{align*}
Defining the function $f^-(x):=f(x0_n)-\min_{a\in\F_q\setminus \{0\}}\big\{f(xa_n)\big\}$, we then get    
\begin{align}\label{expandexpectation}
    \E_{z\sim p^n}\Big[ \sqrt{h_f(z)} \Big]
    &=(1-p)\E_{x\sim p^{n-1}}\Big[ \sqrt{h_{f_0}(x)+f^{-}(x)} \Big]+\frac{p}{q-1}\sum_{\substack{a\in\F_q\setminus \{0\}}}\E_{x\sim p^{n-1}}\Big[ \sqrt{h_{f_a}(x)} \Big].
\end{align}
But applying the Cauchy-Schwarz inequality $\E[\sqrt{gh}]^2\leq \E[g]\E[h]$ and the equality $(a+b)(a-b)=a^2-b^2$, we have
\begin{align*}
    \E[f^{-}]^2&=\E\Big[ {\big(\sqrt{h_{f_0}+f^{-}}-\sqrt{h_{f_0}}\big)^\frac{1}{2}\big(\sqrt{h_{f_0}+f^{-}}+\sqrt{h_{f_0}}\big)}^\frac{1}{2} \Big]^2\\
    &\leq \E\Big[ \sqrt{h_{f_0}+f^{-}}-\sqrt{h_{f_0}}\Big]\E\Big[\sqrt{h_{f_0}+f^{-}}+\sqrt{h_{f_0}} \Big],
\end{align*}
where in the first line we used the fact that $f^{-}$ takes values in $\{0,1\}$, and thus $\sqrt{f^{-}(x)}=f^{-}(x)$ for all $x\in\F_2^{n-1}.$ We can now bound the expected square root of $h_{f_0}+f^{-}$ by
\begin{align*}
    \E\Big[ \sqrt{h_{f_0}+f^{-}}\Big]&\geq\E\Big[ \sqrt{h_{f_0}}\Big]+\frac{\E[f^{-}]^2}{\E\Big[\sqrt{h_{f_0}+f^{-}}+\sqrt{h_{f_0}} \Big]}\\
    &\geq \E\Big[ \sqrt{h_{f_0}}\Big]+\frac{\E[f^{-}]^2}{\E\big[f^{-}\big]+2\E\big[\sqrt{h_{f_0}} \big]},
\end{align*}
where in the last line we used the fact that $\sqrt{a^2+b^2}\leq \sqrt{(a+b)^2}=a+b$, with $a=\sqrt{h_{f_0}}$ and $b=\sqrt{f^{-}}.$
Combining the inequality above with equation (\ref{expandexpectation}), we get
\begin{align*}
    \E\Big[ \sqrt{h_f} \Big]
    &\geq (1-p)\E\Big[ \sqrt{h_{f_0}} \Big]+(1-p)\frac{(E_0-E_{\min})^2}{\E\big[f^{-}\big]+2\E\big[\sqrt{h_{f_0}} \big]}+\frac{p}{q-1}\sum_{\substack{a\in\F_q\setminus \{0\}}}\E\Big[ \sqrt{h_{f_a}} \Big]\\
    &\geq (1-p)\E\Big[ \sqrt{h_{f_0}} \Big]+\frac{p}{q-1}\sum_{\substack{a\in\F_q\setminus \{0\}}}\E\Big[ \sqrt{h_{f_a}} \Big]+\frac{1-p}{2} (E_0-E_{\min})^2,
\end{align*}
where in the last line we used assumption (\ref{assumeEf}) and equation (\ref{expandexpectation}) to get $\E\Big[ \sqrt{h_{f_0}} \Big]\leq \frac{1}{2}$. Applying our induction hypothesis to the functions $\{f_a\}_{a\in\F_q}$, we then have
\begin{align*}
    \E\Big[ \sqrt{h_f} \Big]
    &\geq \frac{(1-p)^2}{2}E_0(1-E_0)+\frac{p}{q-1}\sum_{\substack{a\in\F_q\setminus \{0\}}}\frac{1-p}{2}E_a(1-E_a)+\frac{1-p}{2}(E_0-E_{\min})^2.
\end{align*}
Combining this with equation (\ref{expandvar}), we indeed get
\begin{align*}
            \E\Big[ \sqrt{h_f} \Big]\geq \frac{1-p}{2}\cdot \E[f]\big(1-\E[f]\big).
        \end{align*}
\end{proof}
We now denote the minimum non-zero value of $h_f(x)$ by
\begin{align*}
    \Delta_f:=\min\big\{h_f(z):z\in\F_q^n \textnormal{ such that } h_f(z)\neq 0\big\}.
\end{align*} 
The expectation of $h_f$ can then be bounded as follows.
\begin{theorem}\label{isoperimetry}
For any monotone decreasing function $f:\F_q^n\rightarrow\F_2$ and any noise parameter $p\in[0,1]$, we have
\begin{align*}
            \E\big[h_f\big]&\geq \frac{1-p}{2}\sqrt{\Delta_f}\cdot \E[f]\big(1-\E[f]\big),
    \end{align*}
    where all the expectations are taken with respect to the q-ary $p$-noisy distribution.
\end{theorem}
\begin{proof}
    By Theorem \ref{thmtalagrand} and the Cauchy-Schwarz inequality $\E[\sqrt{g_1g_2}]^2\leq \E[g_1]\E[g_2]$, we have 
    \begin{align*}
            \frac{1-p}{2}\cdot \E[f]\big(1-\E[f]\big)&\leq\E\Big[ \sqrt{h_f} \Big]\\
            &\leq \sqrt{\E\Big[ h_f \Big]\Pr_{x}\big[h_f(x)\neq 0 \big]}.
    \end{align*}
    By definition of $\Delta_f$, we then get
    \begin{align*}
            \frac{1-p}{2}\cdot \E[f]\big(1-\E[f]\big)&\leq \sqrt{\E\Big[ h_f \Big]\cdot\frac{\E\big[h_f\big]}{\Delta_f}}\\
            &=\frac{1}{\sqrt{\Delta_f}}\E\big[h_f\big].
    \end{align*}
\end{proof}

\section{Sharp Transition of the Decoding Error Probability}\label{sectionsharptransition}
In this section, we will show that the decoding error probability of a linear code $C\subseteq\F_q^n$ over the channel qS$C_p$ transitions rapidly from $0$ to $1$ (as a function of $p$). For $q=2$, this was proven by Tillich and Z{\'e}mor in \cite{zemor93application,tillich2000applications2}. We follow a similar approach, and generalize their results to arbitrary field size $q.$ The first building block of our argument is Russo's Lemma, which for $q=2$ first appeared in \cite{margulis1974transition,russo1982transition}. Over arbitrary field size $q$ and for our particular definition of monotonicity, it generalizes as follows.
\begin{lemma}\label{russoFq}
    Let $f:\F_q^n\rightarrow \{0,1\}$ be a monotone decreasing function. Then we have
    \begin{align*}
        \frac{d}{dp}\E_{z\sim p}[f(z)]\leq-\frac{1}{q-1}\E_{z\sim p}[h_f(z)].
    \end{align*}
\end{lemma}

\begin{proof}
    We think of the parameter $p$ as a vector $(p_1,p_2,\dotsc,p_n)$ with $p_i=p$ for all $i\in[n].$ By definition, we have
    \begin{align*}
        \frac{d}{dp}\E_{z\sim p}\big[f(z)\big]&=\sum_{i\in[n]}\frac{d}{dp_i}\E_{z\sim p}\big[f(z)\big]\\
        &=\sum_{i\in[n]}\frac{d}{dp_i} \E_{z\sim p}\Big[ (1-p_i)f(z0_i)+p_i\E_{a\in\{1,2,\dotsc,q-1\}}\big[f(za_i)\big] \Big]\\
        &=\sum_{i\in [n]}\E_{\substack{z\sim p\\a\neq 0}}\big[-f(z0_i)+f(za_i)\big].
    \end{align*}
    Since $f$ is monotone decreasing, we must always have $f(z0_i)\geq f(za_i)$, and thus 
    \begin{align*}
        \frac{d}{dp}\E_{z\sim p}\big[f(z)\big]
        &=-\sum_{i\in [n]}\E_{\substack{z\sim p\\a\neq 0}}\Big[\big|f(z0_i)-f(za_i)\big|\Big].
    \end{align*}
    Considering only the vectors $z$ with $z_i=0$, we then get
    \begin{align*}
        \frac{d}{dp}\E_{z\sim p}\big[f(z)\big]&\leq -\sum_{i\in[n]}\E_{\substack{z\sim p}}\Big[\mathbbm{1}\{z_i=0\}\cdot \E_{a\neq 0}\big[|f(z)-f(za_i)|\big]\Big]\\
        &\leq -\sum_{i\in[n]}\E_{\substack{z\sim p}}\Big[\mathbbm{1}\{z_i=0\}\cdot \frac{1}{q-1}\mathbbm{1}\big\{\exists a\neq0\textnormal{ s.t. }|f(z)-f(za_i)|=1\big\}\Big]\\
        &=-\frac{1}{q-1}\E_{z\sim p}[h_f(z)].
    \end{align*}
\end{proof}
From Lemma \ref{russoFq}, it is clear that for any monotone function $f:\F_q^n\rightarrow\F_2$, any lower bound on $\E[h_f]$ will yield an upper bound on the width of the threshold of $\E[f].$ We thus turn to proving bounds on $\E[h_f]$, for $f$ the indicator function of a successful decoding. For this we will need the following helpful lemma. Recall that for any vectors $a,b\in\F_q^n$, we denote by $d(a,b)$ the Hamming weight of $a-b$.
\begin{lemma}\label{largesupport}
    Suppose $z\in\F_q^n$ and $c\in C$ satisfy 
    \begin{align*}
        d(z,0)\leq d(z,c).
    \end{align*}
    Then we must have
    \begin{align*}
        \big|\textnormal{supp}(c)\setminus \textnormal{supp}(z)\big|\geq\frac{d_{\min}(C)}{q}-d(z,c)+\min_{c'\in C}\{d(z,c')\}.
    \end{align*}
\end{lemma}

\begin{proof}
    For notational simplicity, we define the set 
    \begin{align*}
        S:=\textnormal{supp}(c)\setminus \textnormal{supp}(z)
    \end{align*}
    and the slack quantity
    \begin{align*}
        \nu:=d(z,c)-\min_{c'\in C}\{d(z,c')\}
    \end{align*}
    Our goal is to show that $|S|\geq \frac{d_{\min}}{q}-\nu.$ We first note that since $d(z,c)\geq d(z,0)$, we must have
\begin{align}\label{observeS}
    |S|\geq\big|\{i\in \textnormal{supp}(c)\cap \textnormal{supp}(z): c_i= z_i\}\big|.
\end{align}
We also note that
\begin{align}\label{boundS}
    \big|\{i\in \textnormal{supp}(z)\cap \textnormal{supp}(c): c_i=z_i\}\big|\geq \frac{1}{q-1}\big|\textnormal{supp}(c)\cap \textnormal{supp}(z)\big|-\nu.
\end{align}
This is because for all $\alpha\in\{1,2,\dotsc,q-1\}$, we have
    \begin{align*}
        d(z,\alpha c)=\big|\textnormal{supp}(z)\setminus \textnormal{supp}(c)\big|+\big|\textnormal{supp}(c)\setminus \textnormal{supp}(z)\big|+\big|\{i\in \textnormal{supp}(z)\cap \textnormal{supp}(c):\alpha c_i\neq z_i\}\big|, 
    \end{align*}
    while by averaging there must be some $\alpha\in \{1,2,\dotsc,q-1\}$ such that $$\big|\{i\in \textnormal{supp}(z)\cap \textnormal{supp}(c):\alpha c_i= z_i\}\big|\geq \frac{1}{q-1}\big|\textnormal{supp}(z)\cap \textnormal{supp}(c)\big|.$$
Since $d(z,c)\leq d(z,\alpha c)+\nu$ for every codeword $\alpha c$, we then get equation (\ref{boundS}). With respect to the minimum distance of our code $C$, this gives us
\begin{align*}
    d_{\min}&\leq \textnormal{wt}(c)\\
    &=\big|\textnormal{supp}(c)\setminus \textnormal{supp}(z)\big|+\big|\textnormal{supp}(c)\cap \textnormal{supp}(z)\big|\\
    &\leq |S|+(q-1)(|S|+\nu)\\
    &\leq q|S|+q\nu,
\end{align*}
where in the third line we used equations (\ref{observeS}) and (\ref{boundS}). 
\end{proof}

We are now ready to prove our bound on $\E[h_f]$, for $f$ the indicator function of a successful decoding. Consider the following total order $\prec$ on $\F_q^n$. If $\textnormal{wt}(a)<\textnormal{wt}(b)$, then $a\prec b.$ If $\textnormal{wt}(a)=\textnormal{wt}(b)$ and the support of $a$ comes after the support of $b$ in the lexicographic order, then $a\prec b.$ For completeness' sake (this last point will not appear in our analysis), if $a$ and $b$ have the same support and $a$ comes after $b$ in the full lexicographic order (i.e. the lexicographic order with order $0<1<2<...<q-1$ over $\F_q$), then we say $a\prec b$. Consider the max-likelihood decoder $D^*:\F_q^n\rightarrow C$ defined by
\begin{align}\label{maxdecoderdefn}
    D^*(z):=\min_{c\in C}\{z-c\},
\end{align}
where the comparisons between vectors are taken with respect to the total order $\prec.$ For each codeword $c\in C$, we define the decoding region of $c$ as follows.
\begin{align*}
    \Omega_c:=\{z\in\F_q^n:D^*(z)=c\}.
\end{align*}
\begin{claim}\label{decodersymmetry}
    For all $c\in C$, we have
    \begin{align*}
        \Pr_{z\sim p}[D^*(z+c)=c]&=\Pr_{z\sim p}[z\in\Omega_0].
    \end{align*}
\end{claim}
\begin{proof}
    It is clear that
    \begin{align*}
        \Pr_{z\sim p}[D^*(z)=0]&=\Pr_{z\sim p}[z\in\Omega_0].
    \end{align*}
    Thus it will suffice to show that for any codeword $c\in C$, the map $z\mapsto z+c$ is a bijection between $\Omega_0$ and $\Omega_c.$ But this is indeed the case, as by linearity of $C$ we have
    \begin{align*}
        z\in \Omega_0 &\Longleftrightarrow z\prec z-c'\textnormal{ for all } c'\in C\\
        &\Longleftrightarrow z+c-c \prec z+c-c'\textnormal{ for all } c'\in C\\
        &\Longleftrightarrow z+c\in \Omega_c.
    \end{align*}
\end{proof}
For simplicity, when looking at the $0$ codeword we will drop the subscript and write 
\begin{align*}
\Omega:=\Omega_0.    
\end{align*} 
We will also abuse notation and write $\Delta_\Omega$ and   $h_\Omega$ to mean $\Delta_{\mathbbm{1}_\Omega}$ and $h_{\mathbbm{1}_\Omega}$ respectively.
\begin{lemma}\label{bounddelta}
    Consider any linear code $C\subseteq\F_q^n$. Its corresponding decoding region $\Omega$ satisfies
    \begin{align*}
        \Delta_\Omega \geq \frac{d_{\min}}{q}-3,
    \end{align*}
    where $d_{\min}$ is the minimum distance of $C.$
\end{lemma}
\begin{proof}
    Consider any $z\in\F_q^n$ with $h_{\Omega}(z)\neq 0.$ By definition, the following two conditions must hold.
    \begin{enumerate}[label=(\roman*)]
        \item $D^*(z)=0$,
        \item There exist a codeword $c\in C$ and a coordinate $i\in[n]$ such that $D^*(zc_i)=c.$
    \end{enumerate}
    Our goal will be to show that there are at least $\frac{d_{\min}}{q}-3$ choices of coordinates $i$ where $z_i=0$ and point (ii) above holds. We note that points (i) and (ii) imply that 
    \begin{align}\label{possibleweights}    
    d(z,0)\leq d(z,c)\leq d(z,0)+2.
    \end{align}
    By Lemma \ref{largesupport}, we must then have
    \begin{align}\label{coordsinsupport}
        \big|\textnormal{supp}(c)\setminus \textnormal{supp}(z)\big|\geq\frac{d_{\min}}{q}-2.
    \end{align}
    We now consider two separate cases, depending on the weight of $z-c$.
    
    \textbf{Case 1:} $ d(z,c)\in\big\{d(z,0), d(z,0)+1\big\}.$
    Then for every $j\in \textnormal{supp}(c)\setminus \textnormal{supp}(z)$, we have 
    \begin{align*}
        d(zc_j,c)&=d(z,c)-1\\
        &\leq d(z,0)\\
        &=d(zc_j,0)-1,
    \end{align*} and thus $zc_j\notin \Omega.$ By equation \ref{coordsinsupport}, we thus have $h_{\Omega}(z)\geq \frac{d_{\min}}{q}-2.$

     \textbf{Case 2:} $ d(z,c)= d(z,0)+2.$
    Then for every $j\in \textnormal{supp}(c)\setminus \textnormal{supp}(z)$, we have 
    \begin{align}\label{equalweight}
        d(zc_j,c)&=d(z,c)-1\nonumber\\
        &= d(z,0)+1\nonumber\\
        &=d(zc_j,0).
    \end{align}
    We want to show that for all but one choices of $j\in \textnormal{supp}(c)\setminus \textnormal{supp}(z)$, we have
    \begin{align*}
        zc_j-c\prec zc_j,
    \end{align*}
    or equivalently that the support of $zc_j-c$ comes after the support of $zc_j$ in the lexicographic order. We note that by point (ii) above, there exists a coordinate $i\in[n]$ such that supp$(zc_i-c)$ comes after supp$(zc_i)$ in the lexicographic order. But this means that supp$(z-c)$ must come after supp$(z)$ in the lexicographic order. Define the coordinate
    \begin{align*}
        j^*&:=\min\Big\{j\in \textnormal{supp}(z-c)\setminus\textnormal{supp}(z)\Big\}\\
        &=\min\Big\{j\in \textnormal{supp}(c)\setminus\textnormal{supp}(z)\Big\},
    \end{align*}
    where the minimum is taken over the standard order $1<2<3\dotsc<n$. Then for any coordinate $j>j^*$, supp$(zc_j-c)$ must come after $\textnormal{supp}(zc_j)$ in the lexicographic order. Combining this with equation (\ref{equalweight}), we get that for every coordinate $j\in \textnormal{supp}(c)\setminus\textnormal{supp}(z)$, $j\neq j^*$, we have
    \begin{align*}
        zc_j-c\prec zc_j.
    \end{align*}
    By equation (\ref{coordsinsupport}), there are at least $\frac{d_{\min}}{q}-3$ such coordinates. Thus we indeed have
    \begin{align*}
        h_\Omega(z) \geq \frac{d_{\min}}{q}-3.
    \end{align*}

 \end{proof}
Combining our results from Sections \ref{sectionisoperimetry} and \ref{sectionsharptransition}, we get the following bound on the derivative of the decoding success probability.

\begin{lemma}\label{boundderivative}
    Consider any linear code $C\subseteq\F_q^n$ with minimum distance $d_{\min}\geq 4q$, and any noise parameter $p\in[0,1].$ The decoding region $\Omega$ for the code $C$ satisfies
    \begin{align*}
      \frac{d}{dp}\Pr[z\in\Omega]\leq-\frac{1-p}{4}\cdot \frac{\sqrt{d_{\min}}}{q^{3/2}}\Pr[z\in\Omega]\Big(1-\Pr[z\in\Omega]\Big),
    \end{align*}
where all the probabilities are taken with respect to the q-ary p-biased distribution $z \sim p$.
\end{lemma}    
\begin{proof}
By definition, the decoding region $\Omega$ is monotone decreasing. By Lemma \ref{russoFq} and Theorem \ref{isoperimetry}, we then get
\begin{align*}
    \frac{d}{dp}\Pr_{z\sim p}[z\in\Omega]&\leq-\frac{1}{q-1}\E_{z\sim p}[h_{\Omega}(z)]\\
    &\leq -\frac{1-p}{2(q-1)}\sqrt{\Delta_{\Omega}}\Pr_{z\sim p}[z\in\Omega]\Big(1-\Pr_{z\sim p}[z\in\Omega]\Big).
    \end{align*}
Applying Lemma \ref{bounddelta}, we must thus indeed have 
\begin{align*}
      \frac{d}{dp}\Pr[z\in\Omega]&\leq-\frac{1-p}{2(q-1)}\sqrt{\frac{d_{\min}}{q}-3}\Pr[z\in\Omega]\Big(1-\Pr[z\in\Omega]\Big)\\
      &\leq -\frac{1-p}{4}\cdot \frac{\sqrt{d_{\min}}}{q^{3/2}}\Pr[z\in\Omega]\Big(1-\Pr[z\in\Omega]\Big).
    \end{align*}
\end{proof}
We are now ready to prove our sharp transition result. For convenience, given a fixed code $C$, we denote the probability of a decoding success by
\begin{align*}
    g(p):=\Pr_{z\sim p}[z\in \Omega].
\end{align*}
The theorem below shows that the function $g$ transitions very rapidly from 1 to 0. In spirit, it says that for any noise parameters $p_0<p_1$ that aren't extremely close to each other, either $g(p_0)\approx 1$ or $g(p_1)\approx 0$.
\begin{theorem}\label{gbound}
    Consider any linear code $C\subseteq\F_q^n$ with minimum distance $d_{\min}\geq 4q$, and any noise parameters $0\leq p_0\leq p_1\leq 1.$ Then we have
    \begin{align*}
        g(p_1)\Big(1-g(p_0)\Big)\leq e^{-\frac{1-p_1}{4}\cdot\frac{\sqrt{d_{\min}}}{q^{3/2}}(p_1-p_0)},
    \end{align*}
    where $d_{\min}$ denotes the minimum distance of the code $C.$
\end{theorem}
\begin{proof}
    Define the function
    \begin{align*}
        G(p):=\ln\frac{g(p)}{1-g(p)}.
    \end{align*}
    Then by Lemma \ref{boundderivative}, we have
    \begin{align*}
        \frac{dG}{dp}&=\frac{1}{g(p)\Big(1-g(p)\Big)}\cdot \frac{dg}{dp}\\
        &\leq -\frac{1-p}{4}\cdot\frac{\sqrt{d_{\min}}}{q^{3/2}}.
    \end{align*}
    By the fundamental theorem of calculus, we then have
    \begin{align*}
        G(p_0)-G(p_1)&=-\int_{p_0}^{p_1}\frac{dG}{dp}dp\\
        &\geq\frac{1-p_1}{4}\cdot\frac{\sqrt{d_{\min}}}{q^{3/2}}\Big(p_1-p_0\Big).
    \end{align*}
    By definition of $G$, we thus get
    \begin{align*}
        g(p_1)\Big(1-g(p_0)\Big)&\leq \frac{g(p_1)}{1-g(p_1)}\cdot\frac{1-g(p_0)}{g(p_0)}\\
        &=e^{G(p_1)-G(p_0)}\\
        &\leq e^{-\frac{1-p_1}{4}\cdot\frac{\sqrt{d_{\min}}}{q^{3/2}}(p_1-p_0)}.
    \end{align*}
\end{proof}

\section{Proof of Main Results}\label{sectionmainresult}
In this section, we use our results from Section \ref{sectionsharptransition} to prove Theorems \ref{capacitybridge} and \ref{mainresult}. We first prove a generalization of Theorem \ref{mainresult} below. Taking $\delta=\frac{4q^{3/2}}{(1-p)\sqrt{d_{\min}}}\ln(nL)$ in the following theorem gives Theorem \ref{mainresult}.

\begin{theorem}\label{mainresultformal}
    Let $C\subseteq \F_q^n$ be a linear, $(p,L)$-list decodable code with minimum distance $d_{\min}\geq 4q.$ Then for any $\delta>0$ and any $c\in C$, we have
    \begin{align*}
        \Pr_{z\sim p-n^{-\frac{1}{4}}-\delta}\Big[D^*(c+z)=c \Big]\geq 1-2Le^{-\frac{1-p}{4}\frac{\sqrt{d_{\min}}}{q^{3/2}}\delta}.
    \end{align*}
\end{theorem}

\begin{proof}
    Define the following decoder $D:\F_q^n\rightarrow C.$ Upon seeing a message $m\in\F_q^n$, the decoder $D$ finds all codewords $c\in C$ that satisfy wt$(m-c)\leq pn$, and outputs one of them uniformly at random. The probability of success of this decoder under errors of probability $p-n^{-\frac{1}{4}}$ is bounded by
\begin{align*}
    \Pr_{\substack{z\sim p-n^{-\frac{1}{4}}}}[D(c+z)=c]&\geq \Pr_{\substack{z\sim p-n^{-\frac{1}{4}}}}[\textnormal{wt}(z)\leq pn] \Pr_{\substack{z\sim p-n^{-\frac{1}{4}}}}[D(c+z)=c\big|\textnormal{wt}(z)\leq pn]\\
    &\geq (1-e^{-2\sqrt{n}})\cdot\frac{1}{L}\\
    &\geq \frac{1}{2L},
\end{align*}
where in the second inequality we used Hoeffding's inequality (Lemma \ref{hoeffding}) for the first term, and the fact that $C$ is $(p,L)$-decodable for the second term. Now by Fact \ref{bestdecoder}, the max-likelihood decoder $D^*$ can only have a better decoding probability than $D$, so we have
\begin{align}\label{bdg}
\Pr_{\substack{z\sim p-n^{-\frac{1}{4}}}}[D^*(c+z)=c]
    &\geq\frac{1}{2L}.
\end{align}  
    By Theorem \ref{gbound} and Claim \ref{decodersymmetry}, we then get
    \begin{align*}
        \Pr_{\substack{z\sim p-n^{-\frac{1}{4}}-\delta}}[D^*(c+z)=c]&\geq 1-2L e^{-\frac{1-p}{4}\cdot\frac{\sqrt{d_{\min}}}{q^{3/2}}\delta}.
    \end{align*}
\end{proof}

We then turn to proving Theorem \ref{capacitybridge} from Theorem \ref{mainresult}.

\newtheorem*{capacitybridge}{Theorem \ref{capacitybridge}}
\begin{capacitybridge}
    Let $\{C_n\subseteq\F_{q}^n\}$ be a family of linear codes with rate $1-h_q(p)$, and suppose $\{C_n\}$ achieves list-decoding capacity. If $d_{\min}(C_n)=\omega\Big(\frac{q^3}{(1-p)^2}\Big)$, then $\{C_n\}$ achieves capacity over the q-ary symmetric channel.
\end{capacitybridge}

\begin{proof}
    By Definition \ref{defnlistcapacity}, there exists a function $\epsilon(n)=o(1)$ such that each $C_n$ is $\Big(p-\epsilon_n,\frac{(1-p)^2d_{\min}}{q^3}\Big)$-list decodable. Applying Theorem \ref{mainresultformal} with $\delta = \Big(\frac{q^3}{(1-p)^2d_{\min}}\Big)^{\frac{1}{4}}=o(1)$, we then get
    \begin{align*}
        \Pr_{z\sim p-\epsilon_n-n^{-\frac{1}{4}}-\delta}\Big[D^*(c+z)=c \Big]&\geq 1-\frac{2(1-p)^2d_{\min}}{q^3}\cdot e^{-\frac{1}{4}\big(\frac{(1-p)^2d_{\min}}{q^3}\big)^\frac{1}{4}}\\
        &\geq 1-o(1).
    \end{align*}
\end{proof}

\section*{Acknowledgments}
We thank Hervé Chabanne, Sivakanth Gopi, Anup Rao and Gilles Z{\'e}mor for useful discussions. The work of Francisco Pernice was supported in part by an MIT Jacobs Presidential Fellowship. The work of Oscar Sprumont was supported in part by NSF CCF-2131899, NSF CCF-1813135 and Anna Karlin's Bill and Melinda Gates Endowed Chair. The work of Mary Wootters was supported in part by NSF CCF-2231157 and CCF-2133154.

\appendix
\section{Erasure Channel}\label{aerasure}
In this section, we recall and prove Theorem \ref{capacitythm}.

\newtheorem*{capacitythm}{Theorem \ref{capacitythm}}
\begin{capacitythm}
    Let $C\subseteq\Z_q^n$ be a $(p,L)$-list decodable code with minimum distance $ \omega(\log L).$ Then $C$ admits reliable communication on the qSC$_{p'}$ for $p'=p-\frac{\log n}{\sqrt{n}}.$
\end{capacitythm}

\begin{proof}
Fix any arbitrary sent codeword $c\in C$. For any erasure pattern $z\in \{0,1\}^n$, we define the set of codewords that could be mistaken for $c$ as
    \begin{align*}
        S(z):=\Big\{c'\in C:\restr{c}{\{i\in[n]:z_i=0\}}= \restr{c'}{\{i\in[n]:z_i=0\}}\Big\}.
    \end{align*}
    Our goal will be to show that with high probability over the choice of $z$, $c$ is the only element in $S(z).$ We first note that by Hoeffding's inequality (Lemma \ref{hoeffding}), we have
    \begin{align}\label{farunlikely}
        \Pr_{z\sim p-\frac{\log n}{\sqrt{n}}}[\textnormal{wt}(z)>pn]&<e^{-2\log^2n}\nonumber\\
        &=o(1).
    \end{align}
    
    We also note that by our assumption on the minimal distance of $C$, the probability that any $c'\in C$ be in $S(z)$ can be bounded by 
    \begin{align}\label{probeach}
        \Pr_{z\sim p-\frac{\log n}{\sqrt{n}}}[c'\in S(z)]&=p^{\textnormal{wt}(c+c')}\nonumber\\
        &\leq p^{-\omega(\log nL)}\nonumber\\
        &\leq o(\frac{1}{L}).
    \end{align}
    But since $C$ is $(p,L)$-list decodable, there are at most $L$ codewords $c'\in C$ satisfying wt$(c+c')\leq pn.$ Combining equations (\ref{farunlikely}) and \ref{probeach}) and applying the union bound, we thus get
    \begin{align*}
        \Pr_{z\sim p-\frac{\log n}{\sqrt{n}}}\big[|S(z)|>1\big]&\leq \Pr_{z\sim p-\frac{\log n}{\sqrt{n}}}[\textnormal{wt}(z)>pn]+\sum_{c'\in C:\textnormal{wt}(c+c')\leq pn}\Pr_{z\sim p-\frac{\log n}{\sqrt{n}}}[c'\in S(z)]\\
       &\leq o(1)+L\cdot o(\frac{1}{L})\nonumber\\
       &\leq o(1).
    \end{align*}
\end{proof}

\section{Necessity of distance condition}\label{app:distance}

We construct a linear code $C\subseteq \F_2^n$ of small distance which achieves list-decoding capacity but not q-SC$_p$ capacity. We start with a linear code $C'\subseteq \F_2^n$ which is list-decoding capacity achieving with list size $L' = L_n'.$ Then there is $\eps' = \eps_n'=o(1)$ so that $C'$ is $(p-\eps_n', L_n')$-list decodable. Consider the code $C = \vspan\{e_1, C'\}$, where $(e_1)_i = 1$ if $i=1$ and $(e_1)_i = 0$ otherwise. Fix a $z \in \F_2^n$, and suppose $C'\cap B_{(p-2\eps')n + 1}(z) = \{c_1,\dots, c_t\}$. Since $(p-2\eps')n + 1 < (p-\eps')n$ for all $n$ large enough, we have $t\leq L$ for all $n$ large enough. But then $C\cap B_{(p-2\eps')n} (z) \subseteq \{c_1,\dots, c_t, c_1+e_1, \dots, c_t + e_1\}.$ Hence, setting $\eps = 2\eps'$ and $L = 2L'$, we conclude that $C$ also achieves list-decoding capacity.

However, it's clear that $C$ cannot achieve q-SC capacity. Indeed, if we send some $c\in C$ and the first bit gets corrupted, we can no longer distinguish between $c$ and $c+e_1.$ In other words, the probability of decoding error is bounded below by approximately $p.$

\section{Discussion of the Proofs in \cite{kindarji2010generalization}}\label{appendixmistake}
In this section, we discuss the proof of \cite{kindarji2010generalization} for the claim that the decoding success probability of any $q$-ary linear code with large minimum distance transitions rapidly from $1-o(1)$ to $o(1).$ As far as we can tell, the following two issues suggest that their proof may not be complete. We thank Hervé Chabanne for useful discussions on this subject.

The first issue is that the arguments of \cite{kindarji2010generalization} rely on the following definition of monotonicity: a function $f:\F_q^n\rightarrow\{0,1\}$ is deemed \emph{monotone} if whenever the support of $x\in\F_q^n$ is a subset of the support of $y\in\F_q^n$ and $f(x)=1$, then we must also have $f(y)=1.$ This definition allows for an easier adaptation of $\F_2$ isoperimetric inequalities, and the authors of \cite{kindarji2010generalization} prove sharp thresholds results for functions that are monotone in this sense. They then claim that the non-decoding region of a linear code $C\subseteq\F_q^n$ satisfies this version of monotonicity. Unfortunately, this is not true. For instance, consider the case where $q=3$ and our code is the span of the all-$1$ vector. Then the error string $x$ that has a $1$ in the first $\frac{n}{2}+1$ coordinates and $0$ everywhere else leads to a decoding failure, while the error string $y$ that has a $1$ in the first $\frac{n}{4}+1$ coordinates, a $2$ in the next $\frac{n}{4}+1$ coordinates, and a $0$ everywhere else leads to a decoding success.

The second issue in \cite{kindarji2010generalization} has to do with the neighborhood of boundary points. A critical part of their argument is their claim that any vector that is in the non-decoding region $U_0:=\big\{x\in\F_q^n:\exists c\in C\setminus\{0\}:d(x,c)\leq d(x,0)\big\}$ and has at least one neighbor outside of $U_0$ must have at least $\frac{d}{2}$ such neighbors, where $d$ is the minimum distance of the code. This is true for $q=2$, but it does not hold for larger alphabets. For example, suppose we are working with a field $\F_q$ for some $q\geq n$, and suppose our code is again the span of the all-$1$ vector. Now suppose that the error string $x$ is $x=(0,1,1,2,3,4,\dotsc,n-3,n-2)$. This vector is in the non-decoding region $U_0$, as it is closer to the vector $(1,1,1,...,1)$ than to the 0 vector. It also has a neighbor outside of $U_0$, for instance the vector $(0,0,1,2,3,4,5,...,n-4,n-3,n-2)$. However it only has $2$ such neighbors: the ones obtained by setting the second coordinate or the third coordinate of $x$ to 0. Setting any other coordinate to 0 would yield a vector that is still in $U_0$, as it is equally far from the 0 vector and the all-1 vector.

\bibliographystyle{alpha}
\bibliography{listcapacity}
\end{document}